\documentclass[a4paper,12pt]{article}

\usepackage{amsmath,amssymb,amsthm,url}

\usepackage[margin=1in]{geometry}

\usepackage[ruled,noend,noline,linesnumbered]{algorithm2e}

\usepackage{authblk}

\newtheorem{lemma}{Lemma}[section]
\newtheorem{proposition}[lemma]{Proposition}
\newtheorem{corollary}[lemma]{Corollary}
\newtheorem{theorem}[lemma]{Theorem}

\title{%
  A subquadratic algorithm for the simultaneous conjugacy problem
}

\author[a,b]{Andrej Brodnik\thanks{%
    This work is sponsored in part by the Slovenian
    Research Agency
    (research program P2-0359 and
    research project N2-0053)
  }
}
\author[c,a]{Aleksander Malni\v{c}\thanks{%
    This work is sponsored in part by the Slovenian
    Research Agency
    (research program P1-0285 and
    research projects N1-0062, J1-9108, J1-9110, J1-9187, J1-1694, J1-1695)
  }
}
\author[a]{Rok Po\v{z}ar\thanks{%
    Corresponding author. This work is sponsored in part by the Slovenian
    Research Agency
    (research program P1-0404 and
    research projects N1-0062, J1-9110, J1-9187, J1-1694)
  }
}
\affil[a]{University of Primorska, IAM/FAMNIT}
\affil[b]{University of Ljubljana, FRI}
\affil[c]{University of Ljubljana, PeF}
\begin{document} % ---------------------------------------------------
% --------------------------------------------------------------------

\maketitle

\begin{abstract}
The $d$-Simultaneous Conjugacy problem in the symmetric group $S_n$ asks 
whether there exists a permutation $\tau \in S_n$ such that $b_j = \tau^{-1}a_j \tau$ holds for all  $j = 1,2,\ldots, d$, where $a_1, a_2,\ldots , a_d$ and $b_1, b_2,\ldots , b_d$ are given sequences of permutations in $S_n$. The time complexity of existing algorithms for solving the problem is $O(dn^2)$. We show that for a given  positive integer $d$ the  $d$-Simultaneous Conjugacy problem in $S_n$ can be solved in $o(n^2)$ time.
\end{abstract}

\begin{quote}
  \textbf{Keywords:} canonical labeling, graph isomorphism,
  simultaneous conjugacy problem.
\end{quote}

\section{Introduction}

The {\bf $d$-Simultaneous Conjugacy problem in the symmetric group}  $S_n$ asks whether there exists a permutation  in $S_n$ which simultaneously conjugates  two given $d$-tuples of permutations from $S_n$. More formally,  given two ordered $d$-tuples $a = (a_1, a_2,\ldots , a_d)$ and $b= (b_1, b_2,\ldots , b_d)$ of permutations from $S_n$, is there a permutation $\tau \in S_n$ such that $b_j = \tau^{-1}a_j \tau$ holds for all indices 
$j = 1,2,\ldots, d$? To save words, we shall refer to this problem as $d$-SCP in $S_n$ or even  just as SCP.

This problem arises in many forms in various fields of mathematics and computer science, in particular,  when deciding whether two objects from a given class are structurally equivalent. A brief list includes the following:
in the theory of covering graphs, the  problem of equivalence of covering projections \cite{MNS00}, and moreover, the construction of  all regular covering projections along which a given group of automorphisms lifts \cite{PP17, P19};
in the theory of maps on surfaces, the question whether two oriented maps on a closed surface are combinatorially  isomorphic \cite{MNS02}; in computational group theory,  the problem  whether the centralizer in the symmetric group of a given group is non-trivial \cite{Seress}. Last but not least, in the design of efficient and fast interconnection networks for computer systems, the question of equivalence of permutation networks  also reduces to the SCP \cite{Srid89, Oruc}.

Because of its fundamental importance, the complexity of the $d$-SCP in $S_n$ has been studied since mid seventies \cite{Fontet}.  The problem can be viewed as a special case of the graph isomorphism problem. More precisely, let  $G_a$  be an arc-colored (multi)digraph on the vertex set $[n] = \{1,2,\ldots, n\}$ such that  there is an arc from $u$ to $v$ colored $j$ if and only if  $a_j$ maps $u$ to $v$, for $j=1,2,\ldots, d$. The permutation digraph $G_b$  is defined in a similar fashion. (See Section~\ref{sec:prelim} for a more formal definition.)  The permutation $\tau \in S_n$ that simultaneously conjugates the two tuples  is precisely a color and direction preserving isomorphism from  $G_a$ onto $G_b$ (assuming that permutations in $S_n$ are multiplied from left to right). The graph isomorphism problem   is  hard in general: no polynomial time algorithm is known, nor is the problem known to be NP-complete.  However, there is a recent result due to Babai  \cite{babai-quasi} presenting a quasipolynomial time algorithm.

In our context, things are fundamentally different. Namely,
 when considering the connected components of $G_a$ and $G_b$, the
additional structure imposed by colors is so strong that every color and direction  preserving isomorphism  is uniquely determined by the image of  one  arbitrary vertex. This implies that testing for the existence of such an 
isomorphism can be done in polynomial time. The first algorithm for the $d$-SCP in $S_n$ was proposed in 1977 by Fontet running  in time $O(dn^2)$ \cite{Fontet}. Five years later, the algorithm was independently rediscovered by Hoffmann \cite{Hoff82}.
An important special case of the SCP occurs  when the tuples $a$ and $b$ generate transitive permutation groups or, equivalently, when $G_a$ and $G_b$ are connected.  This restricted problem, referred to as  the transitive SCP, was considered  by Sridhar in 1989 \cite{Srid89}. However, his $O(dn \log(dn))$-time algorithm does not work correctly as we recently showed in \cite{BMP}. Moreover, in the same paper we also  showed that the transitive SCP can be solved in subquadratic time in $n$ at a given $d$; more precisely, we developed an algorithm with the running time $O(n^2  \log d /  \log n + dn\log n)$.

A natural question arises whether  the $d$-SCP in $S_n$ can also be solved in subquadratic time in $n$ at a given $d$. The following main result answers the question affirmatively.
\begin{theorem}\label{thm:main}
Given a positive integer $d$, the  $d$-SCP in the symmetric group $S_n$ can be solved  in $o(n^2)$ time.
\end{theorem}

The main idea  behind our approach is as follows. First, we define two extreme cases depending on the number of connected components  on the one hand, and on the size of individual components on the other hand.  Second, a combination of solutions to these two extreme cases then yields the desired result in general. 

As for the extreme cases, we say that a  connected component is {\bf large}, if it consists of $\Theta(n)$ vertices,
 and   {\bf small} otherwise.  The first extreme case is when a digraph consists  of only large components (and consequently, there are $O(1)$ of them). In the other extreme the digraph consists  of only  small components (and consequently, there are $\omega(1)$ of them). 
 In the first case,  we simply consider each pair of connected components of the same
  size and test whether they are isomorphic by applying the above mentioned subquadratic algorithm  from \cite{BMP}. 
As for the other case, a  different specially tailored approach is used.
 To this end, we present a canonical-labeling-based algorithm that 
 takes $O(dn^2)$ time; however, when both digraphs consist only of small connected components its running time  decreases  to subquadratic in $n$ at a given $d$.

The structure of the paper is the following. Section~\ref{sec:prelim} contains the necessary notation and basic definitions  to make the paper self-contained. In Section~\ref{sec:can} we present a canonical-labeling-based algorithm for the SCP. 
The main theorem is proven in Section~\ref{sec:main}. We conclude
the paper by discussing  some open problems in Section~\ref{sec:conc}.

\section{Permutation digraphs and colour-isomorphism} \label{sec:prelim}

We  establish some notation and terminology used in the paper. For the concepts not defined here see \cite{Diestel}.
%

%Let $S_n$ denote the  group of all permutations of $[n]$, also referred to as the {\bf symmetric group}.
For $i\in [n]$ and $g\in S_n$, we write $i^g$ for the image of $g$ under the permutation $g$ rather than by the more usual $g(i)$.
Let $\sigma = (\sigma_1,\sigma_2,\ldots, \sigma_d)$ be a $d$-tuple of permutations in $S_n$. 
The {\bf permutation digraph} of $\sigma$ is a pair $G_\sigma = (V, A)$, where $V(G_\sigma)= V = [n]$ is  the set of {\bf vertices}, and $A(G_\sigma) = A$ is the set of ordered pairs $(i, \sigma_k)$, $i \in [n], k\in [d]$,  called {\bf arcs}. 
The {\bf size} of $G_\sigma$ is $|V(G_\sigma)|$, while the {\bf degree} of $G_\sigma$ is $|\sigma|$. An arc $e = (i, \sigma_k)$ has its initial vertex $\textrm{ini}(e) = i$,  terminal vertex  $\textrm{ter}(e) = i^{\sigma_k}$, and {\bf color} $c(e) =k$; the vertex  $i^{\sigma_k}$  is also referred  to as the {\bf out-neighbour} of $i$ coloured $k$. The vertices $\textrm{ini}(e)$ and  $\textrm{ter}(e)$ are the end-vertices of $e$. 

A walk from a vertex $v_0$ to a vertex $v_m$ in a permutation digraph $G_\sigma$ is an alternating   sequence  $W= v_0, e_1, v_1, e_2, \ldots, e_m, v_{m}$ of  vertices and arcs in $G$ such that for each  $i\in [m]$, the vertices $v_{i-1}$ and $v_{i}$ are the end-vertices of the arc $e_i$. If for any  two vertices $u$ and $v$ in $G_a$ there is a walk from $u$ to $v$, we say that $G_\sigma$ is {\bf connected}. 
Clearly,  $G_\sigma$  is connected if and only if the tuple $a$ generates a transitive subgroup of $S_n$.
 A subdigraph $H$ of $G_\sigma$ consists of a subset $V(H) \subseteq V(G_\sigma)$  and a subset $A(H) \subseteq A(G_\sigma)$ such that every arc in  $A(H)$ has both end-vertices in $V(H)$.  A walk in a subdigraph $H$ of $G_\sigma$ is a walk in $G_\sigma$ consisting only of arcs from $A(H)$. If $G_{\sigma}$ is not connected, its maximal connected subdigraphs are called the {\bf connected components} of $G_\sigma$. Note that there are no arcs between connected components, and so the components are also permutation digraphs of degree $d$.
 
A {\bf colour-isomorphism} between two permutation digraphs $G_a$ and $G_b$  is a pair $(\phi_V, \phi_A )$ of bijections, where $\phi_V \colon V(G_a) \to V(G_b)$ and $\phi_A \colon A(G_a) \to A(G_b)$ such that $\phi_V(\textrm{ini}(e)) = \textrm{ini}(\phi_A(e))$, $\phi_V(\textrm{ter}(e)) = \textrm{ter}(\phi_A(e))$ and $c(e) = c(\phi_A(e))$ for any arc $e\in A(V_a)$.   
If there is a colour-isomorphism between $G_a$ and $G_b$, we say that $G_a$ and $G_b$ are {\bf colour-isomorphic}, and we write $G_a \cong G_b$.

Let $\mathcal{G}$ be the set of permutation digraphs of size $n$ and degree $d$, and let $L$ be a set of strings over some fixed-sized alphabet. A {\bf labeling function} for $\mathcal{G}$ is a function  $\mathcal{L} \colon \mathcal{G} \to L$. Such a function $\mathcal{L} \colon \mathcal{G} \to L$ is {\bf canonical} whenever for all $G_a, G_b \in \mathcal{G}$  a colour-isomorphism from $G_a$ onto  $G_b$ exists if and only if 
$\mathcal{L}(G_a) = \mathcal{L}(G_b)$. In this case, $\mathcal{L}(G)$ is the  {\bf canonical label} of $G$.

\section{A canonical labeling algorithm}\label{sec:can}

We present an algorithm for finding a canonical label of a permutation digraph $G_a$ of size $n$ and degree $d$ based on publicly known techniques. It runs in $O(dn^2)$ time in general, but in the case when $G_a$  consists of only small connected components its running time decreases to  subquadratic in $n$ at a given $d$.

We first handle the case when $G_a$ is  connected.  
%mogoce kaksna referenca na site graphs ali pa forum na spletu, kjer je ta metoda opisana
For a fixed  $v\in V(G_a)$ we relabel the vertices of $G_a$ in a breadth-first-search order starting at $v$, see the algorithm
  {\sc Relabel}$(G_a, v)$. The out-neighbours of a current vertex are visited in the ascending order of colours of the respective out-going arcs (lines 7-11 in {\sc Relabel}). Let $\gamma_v \colon V \to V$ be the resulting relabeling. 
The  relabelled  digraph induced by $\gamma_v$ is  $G_{a^{v}}$, where $a^{v} = (\gamma_v ^{-1} a_1\gamma_v, \gamma_v ^{-1} a_2\gamma_v, \ldots, \gamma_v ^{-1} a_d\gamma_v )$ (line 12 in {\sc Relabel}).  
%A formal construction of $G_{a^v}$ is given by the algorithm  {\sc Relabel}$(G_a, v)$. 

 \begin{algorithm}[h!]
%\DontPrintSemicolon
\SetNlSty{<text>}{}{:}
%\SetKwFor{For}{for}{do}{}
%\SetKwFor{While}{while}{do}{}
%\SetKwFor{ForEach}{foreach}{do}{}
\NoCaptionOfAlgo
\SetKwComment{Comment}{}{}
\SetCommentSty{<text>}
\SetArgSty{text}
\KwIn{Connected permutation digraph $G_a$  of degree $d$ on $n$ vertices,  $v \in V(G_a)$.}
\KwOut{The permutation digraph $G_{a^v}$.}
Initilize an empty queue $Q$\;
Visited $= \{v\}$\;
$v^{\gamma_v}  = 1$\;
Enqueue $v$ into $Q$\;
\While{$|\textrm{Visited}| \neq n$}{
Dequeue $Q$ into $u$\;
\For{$k\leftarrow 1$ \KwTo $d$}{
\If{$ u^{a_k} \notin \textrm{Visited}$}{
Add $u^{a_k}$ to $\textrm{Visited}$\;
$(u^{a_k})^{\gamma_v}  = |\textrm{Visited}|$\;
Enqueue $u^{a_k}$ into $Q$\;
}
}%end for
}%end while
Let $a^{v} = (\gamma_v ^{-1} a_1\gamma_v, \gamma_v ^{-1} a_2\gamma_v, \ldots, \gamma_v ^{-1} a_d\gamma_v )$\;
\Return The digraph $G_{a^{v}}$\;

\caption{{\bf Algorithm }{\sc Relabel}$(G_a, v)$}
\end{algorithm}
 
The {\bf code} of a permutation digraph  $G_a$  is
$$C(G_a) = 1^{a_1}2^{a_1}\cdots n^{a_1}1^{a_2}2^{a_2}\cdots n^{a_2}\cdots 1^{a_d}2^{a_d} \cdots n^{a_d},$$
which is a string  of length $dn$ over $[n]$ obtained by concatenating, in turn, the images of $1,2,\ldots, n$ under the permutations $a_1, a_2, \ldots, a_d$. %Note that two permutation digraphs with the same code must be equal and hence colour-isomorphic. In contrast,  for each $v\in V(G_a)$, the permutation digraphs $G_a$ and $G_{a^{v}}$ are colour-isomorphic, however, the corresponding codes are not necessarely equal. 
For a connected digraph $G_a$, let $\overline{C}(G_a)$ denote the lexicographically smallest string from among codes $C(G_{a^{v}})$, $v\in V(G_a)$. We now prove that $\overline{C}(G_a)$ is the canonical label of $G_a$. 
\begin{proposition}\label{prop:can-con}
Let $\mathcal{P}^c$ be the set of all connected permutation digraphs of size $n$ and degree $d$, and let $L$ be the set of all strings of length $dn$ over $[n]$.
Then the function  $\mathcal{L}^c \colon \mathcal{P}^c \to L$ defined by 
 $\mathcal{L}^c(G_a) = \overline{C}(G_a)$  is a canonical labeling function for $\mathcal{P}^c$. 
\end{proposition}
\begin{proof}
We first show that if $\overline{C}(G_a) = \overline{C}(G_b)$, then $G_a$ and $G_b$ are colour-isomorphic. Let $u\in V(G_a)$ be a vertex for which the permutation digraph $G_{a^u}$ returned by   {\sc Relabel}$(G_a, u)$ has code $C(G_{a^u}) = \overline{C}(G_a)$. Similarly, let $w\in V(G_b)$ be a vertex for which  the permutation digraph $G_{b^w}$ returned by   {\sc Relabel}$(G_b, w)$ has code $C(G_{b^w}) =  \overline{C}(G_b)$. By assumption, it follows that $C(G_{a^u}) = C(G_{b^w})$. Hence $G_{a^u} = G_{b^w}$, and since  $G_{a^u} \cong G_a$ and $G_{b^w} \cong G_b$ it follows  that  $G_a \cong G_b$.

Conversely, let $f$ be a colour-isomorphism mapping $G_a$  onto $G_b$, and  let $u\in V(G_a)$ be a vertex for which the permutation digraph $G_{a^u}$  returned by   {\sc Relabel}$(G_a, u)$  has code $C(G_{a^u}) = \overline{C}(G_a)$. Next, let $w = f(u)$ and  consider the permutation digraph $G_{b^{w}}$  returned by {\sc Relabel}$(G_b, w)$.
Note that $G_{a^u} = G_{b^{w}},$ and hence $C(G_{a^u}) = C(G_{b^w})$. It remains to prove that $C(G_{b^{w}}) =  \overline{C}(G_b)$. Suppose to the contrary that for some $G_{b^z}$ returned by {\sc Relabel}$(G_b, z)$, the string $C(G_{b^z})$ is  lexicographically smaller than the string $C(G_{b^{w}})$. Consider now the permutation digraph $G_{a^{f^{-1}(z)}}$  returned by   {\sc Relabel}$(G_a, f^{-1}(z))$. Similarly as above,   $C(G_{b^z}) = C(G_{a^{f^{-1}(z)}})$. Since $C(G_{b^z})$ is lexicographically  smaller then $C(G_{b^w}) = C(G_{a^u})$, it follows that  $C(G_{a^{f^{-1}(z)}})$ is lexicographically smaller than $C(G_{a^u})$. A contradiction.
\end{proof}

Next, we  bound the time complexity of computing $\overline{C}(G_a)$.

\begin{lemma}\label{lemma:time}
The canonical label $\overline{C}(G_a)$ of a connected permutation digraph $G_a$ of  size $n$ and degree $d$ can be computed in $O(dn^2)$ time.
\end{lemma}
\begin{proof}
One call of {\sc Relabel}$(G_a, v)$ takes  time $O(dn)$ in order to construct  $G_{a^{v}}$, while its code $C(G_{a^{v}})$ can also be computed in linear time. Since this has to be repeated for each $v\in V(G_a)$, the total running time for constructing the codes  is $O(dn^2)$. Clearly, choosing the lexicographically smallest code does not increase this time bound.
\end{proof}

In the reminder of this section  we consider the case when $G_a$ is not connected. Let us denote the connected components of $G_a$ by $H_1, H_2, \ldots, H_k$, and recall that each such component is a permutation digraph of degree  $d$.
% and  we may assume without loss of generality that its vertex set is equal to $[|V(X_i)|]$. 
Further, let us concatenate the  respective canonical labels  $\overline{C}(H_1), \overline{C}(H_2), \ldots, \overline{C}(H_k)$  in such  an order that the resulting string $C^*(G_a)$ is
  lexicographically smallest. The following result shows that  $C^*(G_a)$ is the canonical label of $G_a$.

%the concatenation $\textrm{CL}(G_a) = \overline{C}(X_{\alpha_1})\overline{C}(X_{\alpha_2})\cdots \overline{C}(X_{\alpha_k})$ be the lexicographically smallest permutation of 
% string obtained by concatenating 
%the  canonical labels $\overline{C}(X_1), \overline{C}(X_2), \ldots, \overline{C}(X_k)$.
%  in an order that produces the lexicographically smallest possible string. 

\begin{theorem}
Let $\mathcal{P}$ be the set of all  permutation digraphs of size $n$ and degree $d$, and let $L$ be the set of all strings of length $dn$ over $[n]$.
Then the function  $\mathcal{L} \colon \mathcal{P} \to L$ defined by 
 $\mathcal{L}(G_a) = C^*(G_a)$  is a canonical labeling function for $\mathcal{P}$.
\end{theorem}

\begin{proof}
The proof follows  from the description of $C^*(G_a)$ above  and Proposition~\ref{prop:can-con}.
\end{proof}

In the next section we will make use of the following result regarding permutation digraphs when all connected components have equal size.

\begin{corollary}\label{cor:cannon}
If  a permutation digraph $G_a$ of size $n$ and $degree$ $d$  consists of precisely $k$ equal-sized  connected components, then its   canonical label $C^*(G_a)$  can be computed in  $O(d n^2  /  k)$ time.
\end{corollary}
\begin{proof}
Let $H_1, H_2, \ldots, H_k$ be the connected components of $G_a$. Since each component $H_i$ is of size  $n/k$ we can find its canonical label $\overline{C}(H_i)$, by  Lemma~\ref{lemma:time},  in $O(d(n/k)^2)$ time. Consequently, 
the total running time for constructing the canonical labels of all components  is $O(dn^2/k)$.
To compute $C^*(G_a)$, all we need to do is to sort these labels. Using radix sort this can be done  in time  $O(dn/k(n/k + k))$,  which obviously does not increase the time bound $O(dn^2/k)$.
\end{proof}

\section{Proof of main result}\label{sec:main}

Recall that a tuple $a$ is  simultaneously  conjugate to a tuple $b$
 if and only if the permutation digraph $G_a$
 is color-isomorphic  to the permutation digraph $G_b$. Before considering the general case when the digraphs $G_a$ and $G_b$ have variable size components, we  deal with two extreme cases, namely, when the digraphs have either only equal-sized small components or only equal-sized large components.
 
\begin{lemma}\label{lemma:small}
Let  $G_a$ and $G_b$ be  permutation digraphs, each of size $n$ and degree $d$, and let both $G_a$ and $G_b$ consist of only small equal-sized connected components.  Then we can test whether $G_a$ and $G_b$ are colour-isomorphic in  time $o(n^2)$ at a given $d$.
\end{lemma}
\begin{proof}
Let $k$ be the number of connected components. By Corollary~\ref{cor:cannon} we can compute the canonical labels of both digraphs and hence perform the isomorphism test  in  time $O(d n^2  /  k)$. Since all components are small, we have $k= \omega(1)$ and consequently $O(d n^2  /  k)= o(n^2)$ at a given $d$.
\end{proof}

%If the permutation digraphs have only small equal-sized connected components, it follows by Corollary~\ref{cor:cannon}  that the isomorphism test can be done in  time $o(n^2)$ at a given $d$. 
% The following result shows that the same  holds for digraphs consisting of only large equal-sized connected components. 

\begin{lemma}\label{lemma:large}
Let  $G_a$ and $G_b$ be  permutation digraphs, each of size $n$ and degree $d$, and let both $G_a$ and $G_b$ consist of only large equal-sized connected components.  Then we can test whether $G_a$ and $G_b$ are colour-isomorphic in  $O(n^2 \log d /  \log n + dn \log n)$ time.
\end{lemma}
\begin{proof}
Let $k$ be the number of connected components. 
Since all components are large, we have $k= O(1)$.  Obviously, at most $k^2 = O(1)$ pairs of  components, each of size $n/k = O(n)$, need to be tested for  isomorphism. By \cite{BMP}, this requires a total of $O(n^2 \log d /  \log n + dn \log n)$ time.
\end{proof}

We are now ready to prove the main result.
% when $G_a$ and $G_b$ consist of variable size components.

%\begin{theorem}\label{thm:graph}
%Let  $G_a$ and $G_b$ be  permutation digraphs each of size $n$  and degree $d$.
%Then we can test whether $G_a$ and $G_b$ are colour-isomorphic in  $o(n^2)$ time at a given $d$.
%\end{theorem}

\begin{proof}[Proof of Theorem~\ref{thm:main}]
Finding the components of $G_a$ and $G_b$ requires  $O(dn)$ time. If $G_a$ and $G_b$ do not
have an equal number of components of  the same size, they are not isomorphic, which can be tested by sorting the sizes of the components in time $o(n^2)$. So, let
 $G_a$ and $G_b$ have $p_i$ components of size $n_i$, $i\in [r]$, where without loss of generality we may assume that components of sizes $n_1, n_2, \ldots, n_j$ are large,  and the remaining ones are small. Obviously, components of different sizes can be tested separately.
% Note that  we must have $p_1, p_2, \ldots, p_j\in O(1)$. 
By Lemma~\ref{lemma:large}, we can test  $p_i = O(1)$ large  components of size $n_i$ for isomorphism   in time $O(n_i^2 \log d / \log(n_i) +  dn_i \log n_i)$. On the other hand, by Corollary~\ref{cor:cannon},  we can test $p_i$ small components of size $n_i$ in time $O(dp_i n_i^2 ) $. Finally, for a large enough constant $c$ the  total time is bounded from above by
$$
c \bigg( \sum_{i=1}^j \frac{d  n_i^2}{ \log n_i}  + \sum_{i=j+1}^rdp_i n_i^2 \bigg) \leq c d \max \bigg\{\frac{  n_1}{ \log(n_1)},  \ldots, \frac{  n_j}{ \log(n_j)}, n_{j+1},\ldots, n_r \bigg\} \sum_{i=1}^r p_in_i.
$$
The max-term  is $o(n)$ since $n_i / \log(n_i) = o(n)$ for each large component as well as $n_i= o(n)$  for each small component. The final result follows as  $\sum_{i=1}^r p_in_i = n$.
\end{proof}

%\begin{proof}[Proof of Theorem~\ref{thm:main}]
%Recall that a tuple $a$ is  simultaneously  conjugate to a tuple $b$  if and only if the permutation digraph $G_a$
% is color-isomorphic  to the permutation digraph $G_b$. Consequently,  the result follows directly from Theorem~\ref{thm:graph}.
% \end{proof}

\section{Concluding remarks}\label{sec:conc}
It remains an open problem whether for a given positive integer $d$ the $d$-SCP in the symmetric group $S_n$ can be solved in a strongly subqudratic time in $n$, that is, in time $O(n^{2-\epsilon})$ for some $\epsilon > 0$. Further, completely unanswered is  the question of the problem's lower bound, except for the trivial one, $\Omega(n)$. The obvious question is whether it can be raised to $\Omega(n \log n)$ reflecting erroneous Sridhar's upper bound, or even to a higher bound by proving conditional lower bounds based on conjectures of hardness for well-studied problems, as it was already done for a number of other problems.

\end{document}